\documentclass{article}
\usepackage[english]{babel}
\usepackage{amsthm,amssymb,stmaryrd}
\usepackage{graphicx}
\usepackage{Prooftree}
\usepackage{Guyboxes}
\usepackage[all]{xy}
\usepackage[colorlinks=true]{hyperref}

\newcommand{\Inn}{\ensuremath{\mathbf{Inn}}}
\newcommand{\Games}{\ensuremath{\mathbf{Gam}}}
\newcommand{\restrict}{\upharpoonright}
\newcommand{\implies}{\Longrightarrow}
\newcommand{\from}{\leftarrow}
\newcommand{\tto}{\Rightarrow}
\newcommand{\leqa}{\unlhd}
\newcommand{\hhole}{\square}
\newcommand{\arena}[1]{{#1}}
\newcommand{\game}[1]{\mathbb{#1}}
\newcommand{\names}{\mathbf{T}}
\newcommand{\intr}[1]{\llbracket #1 \rrbracket}
\newcommand{\bigintr}[1]{\left\llbracket#1\right\rrbracket}

\newlength{\viewht}
\newlength{\viewlift}
\newlength{\viewdp}
\newlength{\viewdrop}

\newcommand{\pview}[1]{
\settoheight{\viewht}{\makebox{$#1$}}
\setlength{\viewlift}{\viewht}%
\addtolength{\viewlift}{-1ex}%
\raisebox{0.3\viewlift}{
  \makebox{$\ulcorner$}}
  \!#1\!
\settoheight{\viewht}{\makebox{$#1$}}
\setlength{\viewlift}{\viewht}%
\addtolength{\viewlift}{-1ex}%
\raisebox{0.3\viewlift}{
  \makebox{$\urcorner$}}
}

\newcommand{\oview}[1]{
\settodepth{\viewdp}{\makebox{$#1$}}
\setlength{\viewdrop}{0.3\viewdp}%
\addtolength{\viewdrop}{0.5ex}%
\raisebox{-\viewdrop}{
  \makebox{$\llcorner$}}
  \!#1\!
\settodepth{\viewdp}{\makebox{$#1$}}
\setlength{\viewdrop}{0.3\viewdp}%
\addtolength{\viewdrop}{0.5ex}%
\raisebox{-\viewdrop}{
  \makebox{$\lrcorner$}}
}

\newtheorem{prop}{\textsc{Proposition}}
\newtheorem{thm}{\textsc{Theorem}}

\author{Pierre \textsc{Clairambault}\\ \texttt{pierre.clairambault@pps.jussieu.fr}\\
PPS --- Universit\'e Paris 7}

\title{Least and Greatest Fixpoints\\ in Game Semantics}
\date{}
\begin{document}

\maketitle
\begin{abstract}
We show how solutions to many recursive arena equations can be computed in a natural way by allowing loops in arenas. 
We then equip arenas with winning functions and total winning strategies. We present two natural winning conditions 
compatible with the loop construction which respectively provide initial algebras and terminal coalgebras for a large 
class of continuous functors. Finally, we introduce an intuitionistic sequent calculus, extended with syntactic 
constructions for least and greatest fixed points, and prove it has a sound and (in a certain weak sense) complete
interpretation in our game model.
\end{abstract}

\section{Introduction}

The idea to model logic by game-theoretic tools can be traced back to the work of Lorenzen \cite{lorenzen1960la}.
The idea is to interpret a formula by a game between two players O and P, O trying to refute the formula and P trying to
prove it. The formula $A$ is then valid if P has a \emph{winning strategy} on the interpretation of $A$. Later, 
Joyal remarked \cite{joyal:rtj} that it is possible to compose strategies in Conway games \cite{conway2001nag} in an associative
way, thus giving rise to the first category of games and strategies. This, along with parallel developments in Linear Logic
and Geometry of Interaction, led to the more recent construction of compositional
game models for a large variety of logics \cite{abramsky1992gaf,mellies2005agf,delataillade2007sot} and programming
languages \cite{hyland2000fap,abramsky2000fap,mccusker2000gaf,abramsky1998fag}.

We aim here to use these tools to model an intuitionistic logic with induction and coinduction. Inductive/coinductive definitions
in syntax have been defined and studied in a large variety of settings, such as linear logic \cite{baelde2007lag}, 
$\lambda$-calculus \cite{abel2000psn} or Martin-L\"of's type theory \cite{dybjer}. Motivations are multiple, but
generally amount to increasing the expressive power of a language without paying the price of  
exponential modalities (as in \cite{baelde2007lag}) or impredicativity (as in \cite{abel2000psn} or \cite{dybjer}).
However, less work has been carried out when it comes to the semantics of such constructions. Of course we
have the famous order-theoretic Knaster-Tarski fixed point theorem \cite{tarski1955ltf}, the nice categorical theory due 
to Freyd \cite{freyd1990acc}, set-theoretic models \cite{dybjer} 
(for the \emph{strictly positive} fragment) or PER-models \cite{loader1997eti}, but it seems they have been ignored by the current 
trend for intensional models (\emph{i.e.} games semantics, GoI \dots). We fix this issue here, showing that (co)induction admits 
a nice game-theoretic model which arises 
naturally if one enriches McCusker's \cite{mccusker2000gaf} work on recursive types with winning functions inspired by parity 
games \cite{santocanale2002fmul}.

In Section 2, we first recall the basic definitions of the Hyland-Ong-Nickau setting of game semantics. Then we sketch 
McCusker's interpretation of recursive types, and show how most of these recursive types can be modelled by means
of loops in the arenas. For this purpose, we define a class of functors called \emph{open functors}, including in particular
all the endofunctors built out of the basic type constructors. We also present a mechanism of winning functions inspired by 
\cite{hyland1997gs}, allowing us to build a category \Games~of games and total winning strategies. In section 3, we present 
$\mu LJ$, the intuitionistic sequent calculus with least and greatest fixpoints that we aim to model. We briefly discuss its 
proof-theoretic properties, then present its semantic counterpart: we show how to build initial algebras and terminal 
coalgebras to most positive open functors. Finally, we use this semantic account of (co)induction to give a sound and 
(weakly) complete interpretation of $\mu LJ$ in \Games.

\section{Arena Games}
\subsection{Arenas and Plays}

We recall briefly the now usual definitions of arena games, introduced in \cite{hyland2000fap}. More detailed accounts
can be found in \cite{mccusker2000gaf,harmer2004igs}. We are interested in games with two participants: Opponent
(O, the \emph{environment}) and Player (P, the \emph{program}). Possible plays are generated by directed graphs
called \emph{arenas}, which are semantic versions of \emph{types} or \emph{formulas}. Hence, a play is a 
sequence of \emph{moves} of the ambient arena, each of them being annotated by a \emph{pointer} to an earlier
move --- these pointers being required to comply with the structure of the arena.
Formally, an \textbf{arena} is a structure $\arena{A}=(M_\arena{A},\lambda_\arena{A},\vdash_\arena{A})$ where:
\begin{itemize}
\item $M_\arena{A}$ is a set of \textbf{moves},
\item $\lambda_\arena{A}:M_\arena{A} \to \{O,P\}\times \{Q,A\}$ is a \textbf{labelling} function indicating whether
a move is an Opponent or Player move, and whether it is a question (Q) or an answer (A). We write $\lambda_\arena{A}^{OP}$
for the projection of $\lambda_\arena{A}$ to $\{O,P\}$ and $\lambda_\arena{A}^{QA}$ for its projection on $\{Q,A\}$. 
$\overline{\lambda_\arena{A}}$ will denote $\lambda_\arena{A}$ where the $\{O,P\}$ part has been reversed.
\item $\vdash_\arena{A}$ is a relation between $M_\arena{A} + \{\star\}$ to $M_\arena{A}$, called \textbf{enabling},
satisfying:
\begin{itemize}
\item $\star\vdash m \implies \lambda_\arena{A}(m) = OQ$;
\item $m \vdash_\arena{A} n \wedge \lambda_\arena{A}^{QA}(n)=A \implies \lambda_\arena{A}^{QA}(m)=Q$;
\item $m\vdash_\arena{A} n \wedge m\neq \star \implies \lambda_\arena{A}^{OP}(m)\neq \lambda_\arena{A}^{OP}(n)$.
\end{itemize}
\end{itemize}
In other terms, an arena is a directed bipartite graph, with a set of distinguished \textbf{initial} moves ($m$ such that
$\star\vdash_\arena{A} m$) and a distinguished set of \textbf{answers} ($m$ such that $\lambda_\arena{A}^{QA}=A$) such that
no answer points to another answer. We now define plays as \textbf{justified sequences} over $\arena{A}$: these are 
sequences $s$ of moves of $\arena{A}$, each non-initial move $m$ in $s$ being equipped with a pointer to an earlier move 
$n$ in $s$, satisfying $n\vdash_\arena{A} m$. In other words, a justified sequence $s$ over $\arena{A}$ is such that
each reversed pointer chain $s_{\phi(0)}\from s_{\phi(1)} \from \dots \from s_{\phi(n)}$ is a path on $\arena{A}$, viewed
as a directed bipartite graph. 

The role of pointers is to allow \emph{reopenings} in plays. Indeed, a path on $\arena{A}$ may be (slightly naively) 
understood as a linear play on $\arena{A}$, and a justified sequence as an interleaving of paths, with possible duplications
of some of them. This intuition is made precise in \cite{harmer2007cci}. When writing justified sequences, we will often omit the 
justification information if this does not cause any ambiguity. $\sqsubseteq$ will denote the prefix ordering on justified 
sequences. If $s$ is a justified sequence on $\arena{A}$, $|s|$ will denote its length.

Given a justified sequence $s$ on $\arena{A}$, it has two subsequences of particular interest: the P-view and O-view.
The view for P (resp. O) may be understood as the subsequence of the play where P (resp. O) only sees his own duplications.
In a P-view, O never points more than once to a given P-move, thus he must always point to the previous move. Concretely,
P-views correspond to branches of B\"ohm trees \cite{hyland2000fap}. Practically, the P-view $\pview{s}$ of $s$ is computed by forgetting everything
under Opponent's pointers, in the following recursive way:
\begin{itemize}
\item $\pview{sm} = \pview{s}m$ if $\lambda_\arena{A}^{OP}(m)=P$;
\item $\pview{sm} = m$ if $\star\vdash_\arena{A} m$ and $m$ has no justification pointer;
\item $\pview{s_1m s_2n} = \pview{s}mn$ if $\lambda_\arena{A}^{OP}(n)=O$ and $n$ points to $m$.
\end{itemize}
The O-view $\oview{s}$ of $s$ is defined dually. Note that in some cases --- in fact if $s$ does not satisfies the 
\emph{visibility condition} introduced below --- $\pview{s}$ and $\oview{s}$ may not be correct justified sequences, since some
moves may have pointed to erased parts of the play. However, we will restrict to plays where this does not happen.
The \textbf{legal sequences} over $\arena{A}$, denoted by $\mathcal{L}_\arena{A}$, are the justified sequences $s$ on $\arena{A}$
satisfying the following conditions:
\begin{itemize}
\item \textbf{Alternation.} If $tmn \sqsubseteq s$, then $\lambda_\arena{A}^{OP}(m)\neq \lambda_\arena{A}^{OP}(n)$;
\item \textbf{Bracketing.} A question $q$ is \textbf{answered} by $a$ if $a$ is an answer and $a$ points to $q$. A question
$q$ is \textbf{open} in $s$ if it has not yet been answered. We require that each answer points to the \emph{pending} question,
\emph{i.e.} the last open question.
\item \textbf{Visibility.} If $tm \sqsubseteq s$ and $m$ is not initial, then if $\lambda_\arena{A}^{OP}(m)=P$ the 
justifier of $m$ appears in $\pview{t}$, otherwise its justifier appears in $\oview{t}$.
\end{itemize} 

\subsection{The cartesian closed category of Innocent strategies}

A \textbf{strategy} $\sigma$ on $\arena{A}$ is a prefix-closed set of even-length legal plays on $\arena{A}$. A strategy
is \textbf{deterministic} if only Opponent branches, \emph{i.e.} $\forall smn,smn'\in \sigma,~n=n'$. Of course, if $\arena{A}$
represents a type (or formula), there are often many more strategies on $\arena{A}$ than programs (or proofs) on this type. 
To address this issue we need \textbf{innocence}. An innocent strategy is a strategy $\sigma$ such that 
\[
sab\in \sigma \wedge t\in \sigma \wedge ta\in \mathcal{L}_\arena{A} \wedge \pview{sa}=\pview{ta} \implies tab\in \sigma
\]

We now recall how arenas and innocent strategies organize themselves into a cartesian closed category. First, we build
the \textbf{product} $\arena{A}\times \arena{B}$ of two arenas $\arena{A}$ and $\arena{B}$:
\begin{eqnarray*}
M_{\arena{A}\times \arena{B}} &=& M_\arena{A} + M_\arena{B}\\
\lambda_{\arena{A}\times \arena{B}} &=& [\lambda_\arena{A},\lambda_\arena{B}]\\
\vdash_{\arena{A}\times \arena{B}} &=&~\vdash_\arena{A} + \vdash_\arena{B}
\end{eqnarray*}

We mention the empty arena $\arena{I} = (\emptyset,\emptyset,\emptyset)$, which will be terminal for the
category of arenas and innocent strategies. We mention as well the arena $\arena{\bot} = (\bullet,\bullet \mapsto OQ,(\star,\bullet))$
with only one initial move, which will be a weak initial object.
We define the \textbf{arrow} $\arena{A}\tto \arena{B}$ as follows:
\begin{eqnarray*}
M_{\arena{A}\tto \arena{B}} &=& M_\arena{A} + M_\arena{B}\\
\lambda_{\arena{A}\tto \arena{B}} &=& [\overline{\lambda_\arena{A}},\lambda_\arena{B}]\\
m \vdash_{\arena{A}\tto \arena{B}} n &\Leftrightarrow& \left\{\begin{array}{l}
				m \neq \star \wedge m \vdash_\arena{A} n \\
				m \neq \star \wedge m \vdash_\arena{B} n \\
				\star\vdash_\arena{B} m \wedge \star\vdash_\arena{A} n \\
				m = \star \wedge \star \vdash_\arena{B} n \end{array}\right.
\end{eqnarray*}

We define composition of strategies by the usual parallel interaction plus hiding mechanism. 
If $\arena{A}$, $\arena{B}$ and $\arena{C}$ are arenas, we define the set of \textbf{interactions} 
$I(\arena{A},\arena{B},\arena{C})$ as the set of justified sequences $u$ over $\arena{A}$, $\arena{B}$
and $\arena{C}$ such that $u_{\restrict_{\arena{A},\arena{B}}}\in \mathcal{L}_{\arena{A}\tto \arena{B}}$,
$u_{\restrict_{\arena{B},\arena{C}}}\in \mathcal{L}_{\arena{B}\tto \arena{C}}$ and 
$u_{\restrict_{\arena{A},\arena{C}}}\in \mathcal{L}_{\arena{A}\tto\arena{C}}$. Then, if $\sigma:\arena{A}\tto \arena{B}$
and $\tau:\arena{B}\tto \arena{C}$, we define parallel interaction:
\[
\sigma || \tau = \{u\in I(\arena{A},\arena{B},\arena{C}) ~|~ u_{\restrict_{\arena{A},\arena{B}}}\in \sigma \wedge u_{\restrict_{\arena{B},\arena{C}}}\in \tau\}
\]
Composition is then defined as $\sigma;\tau = \{u_{\restrict_{\arena{A},\arena{C}}}~|~u\in \sigma||\tau\}$. It is associative 
and preserves innocence (a proof of these facts can be found in \cite{hyland2000fap} or \cite{harmer2004igs}). 
We also define the identity on $\arena{A}$ as the copycat strategy (see \cite{mccusker2000gaf} or \cite{harmer2004igs} for a definition) 
on $\arena{A}\tto \arena{A}$. Thus, there is a category \Inn~which has arenas as objects and innocent strategies on 
$\arena{A}\tto \arena{B}$ as morphisms from $\arena{A}$ to $\arena{B}$. In fact, this category is cartesian closed, the cartesian 
structure given by the arena product above and the exponential closure given by the arrow construction.
This category is also equipped with a weak coproduct $\arena{A}+ \arena{B}$ \cite{mccusker2000gaf}, which is constructed as follows:
\begin{eqnarray*}
M_{\arena{A}+\arena{B}} &=& M_\arena{A} + M_{\arena{B}} + \{q,L,R\}\\
\lambda_{\arena{A}+\arena{B}} &=& [\lambda_\arena{A},\lambda_\arena{B},q\mapsto OQ,L\mapsto PA, R\mapsto PA]\\
m \vdash_{\arena{A}+\arena{B}} n &\Leftrightarrow&
	\left\{\begin{array}{l}m,n\in M_\arena{A} \wedge m\vdash_\arena{A} n \\
			       m,n\in M_\arena{B} \wedge m\vdash_\arena{B} n \\
			       m=\star \wedge n=q\\
			       (m=q \wedge n=L) \vee (m=q \wedge n=R)\\
			       (m=L \wedge \star\vdash_\arena{A} n) \vee (m=R \wedge \star\vdash_\arena{B} n)\end{array}\right.
\end{eqnarray*}

\label{arenasplays}
\subsection{Recursive types and Loops}

Let us recall briefly the interpretation of recursive types in game semantics, due to McCusker \cite{mccusker2000gaf}.
Following \cite{mccusker2000gaf}, we first define an ordering $\leqa$ on arenas as follows. For two arenas $\arena{A}$
and $\arena{B}$, $\arena{A}\leqa \arena{B}$ iff
\begin{eqnarray*}
M_\arena{A} &\subseteq& M_\arena{B}\\
\lambda_\arena{A} &=& {\lambda_\arena{B}}_{\restrict_{M_\arena{A}}}\\
\vdash_\arena{A} &=& ~\vdash_\arena{B}\cap~(M_\arena{A}+\{\star\}\times M_\arena{A})\\
\end{eqnarray*}
This defines a (large) dcpo, with least element $\arena{I}$ and directed sups given by the componentwise union.
If $F:\Inn\to \Inn$ is a functor which is continuous with respect to $\leqa$, we 
can find an arena $\arena{D}$ such that $\arena{D}=F(\arena{D})$ in the usual way by setting
$D=\bigsqcup_{n=0}^{\infty}F^n(\arena{I})$. McCusker showed \cite{mccusker2000gaf} that when the functors are \emph{closed}
(\emph{i.e.} their action can be internalized as a morphism $(A\tto B)\to (F A \tto F B)$), and 
when they preserve inclusion and projection morphisms (\emph{i.e.} partial copycat strategies) corresponding to $\leqa$, 
this construction defines \emph{minimal invariants} \cite{freyd1990acc}.
Note that the crucial cases of these constructions are the functors built out of the product, sum and
function space constructions.

We give now a concrete and new (up to the author's knowledge) description of a large class of continuous functors, that
we call \textbf{open functors}. These include all the functors built out of the basic constructions, and
allow a rereading of recursive types, leading to the model of (co)induction. 
\subsubsection{Open arenas.}Let $\names$ be a countable set of names. 
An \textbf{open arena} is an arena $\arena{A}$ with distinguished question moves called \textbf{holes}, 
each of them labelled by an element of $\names$. We denote by $\hhole_X$ the holes annotated by $X\in \names$.
We will sometimes write $\hhole^P_X$ to denote a hole of Player polarity, or $\hhole^O_X$ to denote a hole
of Opponent polarity. If $\arena{A}$ has holes labelled by $X_1,\dots,X_n$, we denote it by $\arena{A}[X_1,\dots,X_n]$. 
By abuse of notation, the corresponding open functor we are going to build will be also denoted by 
$\arena{A}[X_1,\dots,X_n] : (\Inn \times \Inn^{op})^n \to \Inn$.

\subsubsection{Image of arenas.}If $\arena{A}[X_1,\dots,X_n]$ is an open arena and
$\arena{B}_1,\dots,\arena{B}_n,\arena{B}'_1,\dots,\arena{B}'_{n}$
are arenas (possibly open as well), we build a new arena $\arena{A}(\arena{B}_1,\arena{B}'_1,\dots,\arena{B}_n,\arena{B}'_n)$ by replacing
each occurrence of $\hhole^P_{X_i}$ by $\arena{B}_{i}$ and each occurrence of $\hhole^O_{X_i}$
by $\arena{B}'_{i}$. More formally:
\begin{eqnarray*}
M_{\arena{A}(\arena{B}_1,\arena{B}'_1,\dots,\arena{B}_n,\arena{B}'_n)}&=&(M_\arena{A}\setminus\{\hhole_{X_1},\dots,\hhole_{X_n}\})+\sum_{i=1}^{n}{(M_{\arena{B}_i}+M_{\arena{B}'_i})}\\
\lambda_{\arena{A}(\arena{B}_1,\arena{B}'_1,\dots,\arena{B}_n,\arena{B}'_n)}&=&[\lambda_\arena{A},\lambda_{\arena{B}_1},\overline{\lambda_{\arena{B}'_1}},\dots,\lambda_{\arena{B}_n},\overline{\lambda_{\arena{B}'_n}}]\\
m\vdash_{\arena{A}(\arena{B}_1,\arena{B}'_1,\dots,\arena{B}_n,\arena{B}'_{n})} p &\Leftrightarrow& \left\{\begin{array}{l}
m \vdash_\arena{A} \hhole^P_{X_i} \wedge \star\vdash_{\arena{B}_{i}} p\\
m \vdash_\arena{A} \hhole^O_{X_i} \wedge \star\vdash_{\arena{B}'_i} p\\
\star\vdash_{\arena{B}_{i}} m \wedge \hhole^P_{X_i} \vdash_\arena{A} p\\
\star\vdash_{\arena{B}'_i}  m \wedge \hhole^O_{X_i} \vdash_\arena{A} p\\
m \vdash_{\arena{B}_i} p\\
m \vdash_{\arena{B}'_i} p\\
m \vdash_{\arena{A}} p
\end{array}\right.
\end{eqnarray*}
Note that in this definition, we assimilate all the moves sharing the same hole \emph{label} $\hhole_{X_i}$ and with 
the same polarity. This helps to clarify notations, and is justified by the fact that we never need to distinguish 
moves with the same hole label, apart from when they have different polarity.

\subsubsection{Image of strategies.} If $\arena{A}$ is an arena, we will, by abuse of notation, denote by $I_\arena{A}$ both
the set of initial moves of $\arena{A}$ and the subarena of $\arena{A}$ with only these moves. Let $\arena{A}[X_1,\dots,X_n]$ 
be an open arena, $B'_1,B_1,\dots,B'_n,B_n$ and $C'_1,C_1,\dots,C'_n,C_n$ be arenas.
Consider the application $\xi$ defined on moves as follows:
\[
\xi(x)=\left\{\begin{array}{ll}	\hhole_{X_i}&\txt{if $x\in \bigcup_{i\in \{1,\dots,n\}}{(I_{B'_i}\cup I_{B_i} \cup I_{C'_i} \cup I_{C_i})}$}\\
				x & \txt{otherwise}
		\end{array}\right.
\]
and then extended recursively to an application $\xi^*$ on legal plays as follows:
\[
\xi^*(sa) = \left\{\begin{array}{l}
			\xi^*(s)\txt{ if $a$ is a non-initial move of $B_i,B'_i,C_i$ or $C'_i$}\\
			\xi^*(s)\xi(a)\txt{ otherwise}
		     \end{array}\right.
\]

$\xi^*$ erases moves in the inner parts of $B'_i,B_i,C'_i,C_i$ and agglomerates all the initial 
moves back to the holes. This way we will be able to compare the resulting play with the identity on $\arena{A}[X_1,\dots,X_n]$.
Now, if $\sigma_i:\arena{B}_{i}\to \arena{C}_{i}$ and $\tau_i:\arena{C}'_{i}\to \arena{B}'_{i}$ are strategies, 
we can now define the action of open functors on them by stating:
\[
s\in \arena{A}(\sigma_1,\tau_1,\dots,\sigma_n,\tau_n) \Leftrightarrow
\left\{\begin{array}{l}
\forall i\in \{1,\dots,n\},~s_{\restrict_{B_i\tto C_i}}\in \sigma_i\\
\forall i\in \{1,\dots,n\},~s_{\restrict_{C'_i\tto B'_i}}\in \tau_i\\
\xi^*(s) \in id_{\arena{A}[X_1,\dots,X_n]}
\end{array}\right.
\]

\begin{prop}
For any $\arena{A}[X_1,\dots,X_n]$, this defines a functor 
$\arena{A}[X_1,\dots,X_n]:(\Inn \times \Inn^{op})^n \to \Inn$, which is monotone and continuous with respect to $\leqa$.
\end{prop}
\begin{proof}[Proof sketch] Preservation of identities and composition are rather direct. A little care 
is needed to show that the resulting strategy is innocent: this relies on two facts: 
First, for each Player move the three definition cases are mutually exclusive. Second, a P-view of 
$s\in  \arena{A}(\sigma_1,\tau_1,\dots,\sigma_n,\tau_n)$ is (essentially) an initial copycat
appended with a P-view of one of $\sigma_i$ or $\tau_i$, hence the P-view of $s$ determines uniquely the P-view presented 
to one of $\sigma_i$, $\tau_i$ or $id_{\arena{A}[X_1,\dots,X_n]}$.
\end{proof}

\paragraph{Example.} Consider the open arena $\arena{A}[X] = \hhole_X \tto \hhole_X$. For any arena $\arena{B}$, 
we have $\arena{A}(\arena{B}) = \arena{B}\tto \arena{B}$ and for any $\sigma : B_1 \to C_1$ and $\tau : C_2 \to B_2$,
we have $\arena{A}(\sigma,\tau) = \tau \tto \sigma : (B_2\tto B_1) \to (C_2\tto C_1)$, the strategy which precomposes
its argument by $\tau$ and postcomposes it by $\sigma$.

\subsubsection{Loops for recursive types.} Since these open functors are monotone and continuous with respect to $\leqa$, 
solutions to their corresponding recursive equations can be obtained by computing the infinite expansion of arenas (\emph{i.e.}
infinite iteration of the open functors).
However, for a large subclass of the open functors, this
solution can be expressed in a simple way by replacing holes with a loop up to the initial moves.
Suppose $\arena{A}[X_1,\dots,X_n]$ is an open functor, and $i$ is such that $\hhole_{X_i}$ appears only in non-initial, positive
positions in $\arena{A}$. Then we define an arena $\mu X_i.\arena{A}$ as follows:
\begin{eqnarray*}
M_{\mu X_i.\arena{A}} &=& (M_\arena{A}\setminus\hhole_{X_i})\\
\lambda_{\mu X_i.\arena{A}} &=& {\lambda_\arena{A}}_{\restrict_{M_{\mu X_i.\arena{A}}}}\\
m\vdash_{\mu X_i.\arena{A}} n &\Leftrightarrow& \left\{\begin{array}{l}
m\vdash_\arena{A} n\\
m\vdash_\arena{A} \hhole_{X_i} \wedge \star\vdash_\arena{A} n\end{array}\right.
\end{eqnarray*}

A simple argument ensures that the obtained arena is isomorphic to the one obtained by iteration of the functor. For this
issue we take inspiration from Laurent \cite{laurent2005cit} and prove a theorem stating that two arenas are isomorphic in 
the categorical sense if and only if their set of paths are isomorphic. A \textbf{path} in $\arena{A}$ is a sequence of 
moves $a_1,\dots,a_n$ such that for all $i\in \{1,\dots,n-1\}$ we have $a_i \vdash_\arena{A} a_{i+1}$. A \textbf{path isomorphism}
between $\arena{A}$ and $\arena{B}$ is a bijection $\phi$ between the set of paths of $\arena{A}$ and the set of paths 
on $\arena{B}$ such that for any non-empty path $p$ on $\arena{A}$, $\phi(ip(p)) = ip(\phi(p))$ (where $ip(p)$ denotes
the immediate prefix of $p$). We have then the theorem:

\begin{thm}Let $\arena{A}$ and $\arena{B}$ be two arenas. They are categorically isomorphic if and only if 
there is a path isomorphism between their respective sets of paths.
\label{iso}
\end{thm}

Now, it is clear by construction that, if $\arena{A}[X]$ is an open functor such that $\hhole_X$ appears only in non-initial positive
positions in $\arena{A}$, the set of paths of $\bigsqcup_{n=0}^{\infty}\arena{A}^n(\arena{I})$
and of $\mu X.\arena{A}$ are isomorphic. Therefore $\mu X.\arena{A}$ is solution of the recursive equation $X = \arena{A}(X)$, and 
when $\arena{A}[X]$ is closed and preserves inclusions and projections, $\mu X.\arena{A}$ defines as well a minimal invariant for $\arena{A}[X]$. 
But in fact, we have the following fact:

\begin{prop}If $\arena{A}[X]$ is an open functor, then it is closed and preserves inclusions and
projections. Hence $\mu X. \arena{A}$ is a minimal invariant for $\arena{A}[X]$.
\end{prop}

This interpretation of recursive types as loops
preserves finiteness of the arena, and as we shall see, allows to easily express the winning conditions necessary
to model induction and coinduction.

\label{loops}
\subsection{Winning and Totality}

A \textbf{total strategy} on $\arena{A}$ is a strategy $\sigma:\arena{A}$ such that for all $s\in \sigma$, if there
is $a$ such that $sa\in \mathcal{L}_\arena{A}$, then there is $b$ such that $sab\in \sigma$. In other words, 
$\sigma$ has a response to any legal Opponent move. This is crucial to interpret logic because the interpretation
of proofs in game semantics always gives total strategies: this is a counterpart in semantics to the cut
elimination property in syntax. To model induction and coinduction in logic, we must therefore restrict to total 
strategies. However, it is well-known that the class of total strategies is not closed under composition, because 
an infinite chattering can occur in the hidden part of the interaction. This is analogous to the fact that in 
$\lambda$-calculus, the class of strongly normalizing terms is not closed under application: $\delta = \lambda x.x x$ 
is a normal form, however $\delta \delta$ is certainly not normalizable. This problem is discussed in \cite{abramsky1996sii,hyland1997gs}
and more recently in \cite{HCtotality}. We take here the solution of \cite{hyland1997gs}, and equip arenas with winning functions: 
for every infinite play we choose a loser, hence restricting to winning strategies has the effect of blocking infinite chattering.

The definition of legal plays extends smoothly to infinite plays. Let $\mathcal{L}_\arena{A}^\omega$ denote the set of infinite
legal plays over $\arena{A}$. If $\overline{s}\in \mathcal{L}_\arena{A}^\omega$, we say that $\overline{s}\in \sigma$ when
for all $s\sqsubset \overline{s}$, $s\in \sigma$. We write 
$\overline{\mathcal{L}_\arena{A}} = \mathcal{L}_\arena{A} + \mathcal{L}_\arena{A}^\omega$. A \textbf{game} will be a pair
$\game{A} = (\arena{A},\mathcal{G}_\arena{A})$ where $\arena{A}$ is an arena, and $\mathcal{G}_\arena{A}$ is a 
function from infinite \textbf{threads} on $\arena{A}$ (\emph{i.e.} infinite legal plays with exactly one initial move) 
to $\{W,L\}$. The winning function $\mathcal{G}_\arena{A}$ extends naturally to potentially finite threads by setting, 
for each finite $s$:
\[
\mathcal{G}_\arena{A}(s) = \left\{\begin{array}{ll}W&\textrm{if $|s|$ is even ;}\\
						     L&\textrm{otherwise.}\end{array}\right.
\]
Finally, $\mathcal{G}_\arena{A}$ extends to legal plays by saying that $\mathcal{G}_\arena{A}(s)=W$ iff 
$\mathcal{G}_\arena{A}(t)=W$ for every thread $t$ of $s$. By abuse of notation, we keep the same notation
for this extended function. 
The constructions on arenas presented in section \ref{arenasplays} extend to constructions on games as follows:
\begin{itemize}
\item $\mathcal{G}_{\arena{A}\times \arena{B}}(s) = [\mathcal{G}_\arena{A},\mathcal{G}_\arena{B}]$ (indeed, a thread
on $\arena{A}\times \arena{B}$ is either a thread on $\arena{A}$ or a thread on $\arena{B}$) ;
\item $\mathcal{G}_{\arena{A}+\arena{B}}(s) = W$ iff all threads of $s_{\restrict_\arena{A}}$ are winning for 
$\mathcal{G}_\arena{A}$ and all threads of $s_{\restrict_\arena{B}}$ are winning for $\mathcal{G}_\arena{B}$.
\item $\mathcal{G}_{\arena{A}\Rightarrow \arena{B}}(s)=W$ iff if all threads of $s_{\restrict_\arena{A}}$ are winning
for $\mathcal{G}_\arena{A}$, then $\mathcal{G}_\arena{B}(s_{\restrict_\arena{B}})=W$.
\end{itemize}

It is straightforward to check that these constructions commute with the extension of winning functions from infinite 
threads to potentially infinite legal plays. 
We now define \textbf{winning strategies} $\sigma:\game{A}$ as innocent
strategies $\sigma:\arena{A}$ such that for all $s\in \sigma$, $\mathcal{G}_\arena{A}(s)=W$. Now, the following
proposition is satisfied:

\begin{prop}Let $\sigma:\game{A}\Rightarrow \game{B}$ and $\tau:\game{B}\Rightarrow \game{C}$ be two
total winning strategies. Then $\sigma;\tau$ is total winning.
\end{prop}
\begin{proof}[Proof sketch.]
If $\sigma;\tau$ is not total, there must be infinite $s$ in their parallel interaction $\sigma||\tau$, such that
$s_{\restrict_{\arena{A},\arena{C}}}$ is finite. By switching, we have in fact $|s_{\restrict_\arena{A}}|$ even
and $|s_{\restrict_\arena{C}}|$ odd. Thus $\mathcal{G}_\arena{A}(s_{\restrict_\arena{A}})=W$ and 
$\mathcal{G}_\arena{C}(s_{\restrict_\arena{C}})=L$. We reason then by disjunction of cases. Either 
$\mathcal{G}_\arena{B}(s_{\restrict_\arena{B}})=W$ in which case 
$\mathcal{G}_{\arena{B}\Rightarrow \arena{C}}(s_{\restrict_{\arena{B},\arena{C}}})=L$ and $\tau$ cannot
be winning, or $\mathcal{G}_\arena{B}(s_{\restrict_\arena{B}})=L$ in which case 
$\mathcal{G}_{\arena{A}\Rightarrow \arena{B}}(s_{\restrict_{\arena{A},\arena{B}}})=L$ and $\sigma$ cannot
be winning. Therefore $\sigma;\tau$ is total.

$\sigma;\tau$ must be winning as well. Suppose there is $s\in \sigma;\tau$ such that 
$\mathcal{G}_{\arena{A}\Rightarrow \arena{C}}(s)=L$. By definition of $\mathcal{G}_{\arena{A}\Rightarrow\arena{C}}$,
this means that $\mathcal{G}_\arena{A}(s_{\restrict_\arena{A}})=W$ and $\mathcal{G}_\arena{C}(s_{\restrict_\arena{C}})=L$.
By definition of composition, there is $u\in \sigma||\tau$ such that $s=u_{\restrict_{\arena{A},\arena{C}}}$. 
But whatever the value of $\mathcal{G}_\arena{B}(u_{\restrict_\arena{B}})$ is, one of $\sigma$ or $\tau$ is
losing. Therefore $\sigma;\tau$ is winning.
\end{proof}

It is clear from the definitions that all plays in the identity are winning. It is also clear that all the structural 
morphisms of the cartesian closed structure of \Inn~are winning (they are essentially copycat strategies), thus
this defines a cartesian closed category \Games~of games and innocent total winning strategies. 

\section{Fixpoints}
\subsection{$\mu LJ$: an intuitionistic sequent calculus with fixpoints}
\subsubsection{Formulas.}$S~::= ~S\tto T~ |~ S\vee T~ |~ S\wedge T ~|~ \mu X.T ~|~ \nu X.T ~|~ X ~|~ \top ~|~ \bot$\\
A formula $F$ is \textbf{valid} if for any subformula of $F$ of the form $\mu X.F'$, 
\begin{itemize}
\item[\emph{(1)}] $X$ appears only positively in $F'$,
\item[\emph{(2)}] $X$ does not appear at the root of $F'$ (\emph{i.e.} $X$ appears at least under a $\vee$ or a $\tto$ in the abstract
syntax tree of $F'$).
\end{itemize}
\emph{(2)} corresponds to the restriction to arenas where loops allow to express recursive types, whereas \emph{(1)} is the usual positivity condition.
We could of course hack the definition to get rid of these restrictions, but we choose not to obfuscate the treatment for an extra 
generality which is neither often considered in the literature, nor useful in practical examples of (co)induction.

\subsubsection{Derivation rules.} We present the rules with the usual dichotomy.
\boxit[Identity group]{
\[
\prooftree
\justifies
A\vdash A
\using ax
\endprooftree
~~~~~~~~~~~~~
\prooftree
\Gamma\vdash A~~~~~\Delta,A\vdash B
\justifies
\Gamma,\Delta\vdash B
\using Cut
\endprooftree
\]
}
\boxit[Structural group]{
\[
\prooftree
\Gamma,A,A\vdash B
\justifies
\Gamma,A\vdash B
\using C
\endprooftree
~~~~~~~~~~~~~
\prooftree
\Gamma \vdash B
\justifies
\Gamma,A\vdash B
\using W
\endprooftree
~~~~~~~~~~~~~
\prooftree
\Gamma,A,B,\Delta\vdash C
\justifies
\Gamma,B,A,\Delta\vdash C
\using \gamma
\endprooftree
\]
}
\boxit[Logical group]{
\[
\prooftree
\Gamma,A\vdash B
\justifies 
\Gamma \vdash A\tto B
\using \tto_r
\endprooftree
~~~~~~~~~~
\prooftree
\Gamma\vdash A~~~~~\Delta,B\vdash C
\justifies 
\Gamma,\Delta, A\tto B \vdash C
\using \tto_l
\endprooftree
~~~~~~~~~~
\prooftree
\justifies
\Gamma,\bot\vdash A
\using \bot_l
\endprooftree
~~~~~~~~~~
\prooftree
\justifies
\Gamma \vdash \top
\using \top_r
\endprooftree
\]\vspace{5mm}
\[
\prooftree
\Gamma\vdash A ~~~~ \Gamma\vdash B
\justifies 
\Gamma \vdash A\wedge B
\using \wedge_r
\endprooftree
~~~~~~~~~~~~
\prooftree
\Gamma,A\vdash C
\justifies 
\Gamma,A\wedge B\vdash C
\using \overleftarrow{\wedge_l}
\endprooftree
~~~~~~~~~~~~
\prooftree
\Gamma,B\vdash C
\justifies 
\Gamma,A\wedge B \vdash C
\using \overrightarrow{\wedge_l}
\endprooftree
\]\vspace{5mm}
\[
\prooftree
\Gamma \vdash A
\justifies 
\Gamma \vdash A \vee B
\using \overleftarrow{\vee_r}
\endprooftree
~~~~~~~~~~~~
\prooftree
\Gamma \vdash B
\justifies 
\Gamma \vdash A \vee B
\using \overrightarrow{\vee_r}
\endprooftree
~~~~~~~~~~~~
\prooftree
\Gamma,A \vdash C~~~~~~\Delta,B\vdash C
\justifies
\Gamma,\Delta,A\vee B\vdash C
\using \vee_l
\endprooftree
\]
}
\boxit[Fixpoints]{
\[
\prooftree
\Gamma\vdash T[\mu X.T/X]
\justifies 
\Gamma \vdash \mu X.T
\using \mu_r
\endprooftree
~~~~~
\prooftree
T[A/X]\vdash A
\justifies 
\mu X.T \vdash A
\using \mu_l
\endprooftree
~~~~~
\prooftree
T[\nu X.T/X] \vdash B
\justifies 
\nu X.T \vdash B
\using \nu_l
\endprooftree
~~~~~
\prooftree
A\vdash T[A/X]
\justifies 
A\vdash \nu X.T
\using \nu_r
\endprooftree
\]
}

Note that the $\mu_l$, $\nu_l$ and $\nu_r$ rules are not relative to any context. In fact, the general rules
with a context $\Gamma$ at the left of the sequent are derivable from these ones (even if, for $\mu_l$ and $\nu_r$,
the construction of the derivation requires an induction on $T$), and we stick with the present ones to clarify the game model.
Cut elimination on the $\tto,\wedge,\vee$ fragment is the same as usual. For the reduction of $\mu$ and $\nu$, we 
need an additional rule to handle the unfolding of formulas.
For this purpose, we add a new rule $[T]$ for each type $T$ with free variables. This method can already be found in 
\cite{abel2000psn} for strictly positive functors: no type variable appears on the left of an implication.
From now on, $T[A/X]$ will be abbreviated $T(A)$. This notation implies that, unless otherwise stated, $X$ will
be the variable name for which $T$ is viewed as a functor. In the following rules, $X$ appears only positively in $T$ 
and only negatively in $N$:

\boxit[Functors]{
\[
\prooftree
A\vdash B
\justifies
T(A)\vdash T(B)
\using [T]
\endprooftree
~~~~~~~~~~~~~
\prooftree
A\vdash B
\justifies
N(B)\vdash N(A)
\using [N]
\endprooftree
\]
}
The dynamic behaviour of this rule is to locally perform the unfolding. We give some of the reduction
rules. These are of two kinds: the rules for the elimination of $[T]$, and the cut elimination rules. Here are the
main cases:

{
\small{}
\[
\prooftree
        \prooftree
                \pi
                \justifies
                A\vdash B
        \endprooftree
        \justifies
        T \vdash T
        \using [T] (X\not\in FV(T))
\endprooftree
\leadsto
\prooftree
\justifies
T\vdash T
\using ax
\endprooftree
~~~~~~~~~~~~~~
\prooftree
        \prooftree
                \pi
                \justifies
                A\vdash B
        \endprooftree
        \justifies
        A\vdash B
        \using [X]
\endprooftree
\leadsto
\prooftree
        \pi
        \justifies
        A\vdash B
\endprooftree
\]\vspace{5mm}
\[
\prooftree
        \prooftree
                \pi
                \justifies
                A\vdash B
        \endprooftree
        \justifies 
        N(A)\tto T(A) \vdash N(B)\tto T(B)
        \using [N\tto T]
\endprooftree
\leadsto
\prooftree
        \prooftree
                \prooftree
                        \prooftree
                                \pi
                                \justifies
                                A\vdash B
                        \endprooftree
                        \justifies
                        N(B)\vdash N(A)
                        \using [N]
                \endprooftree
                ~~~~~
                \prooftree
                        \prooftree
                                \pi
                                \justifies
                                A\vdash B
                        \endprooftree
                        \justifies
                        T(A)\vdash T(B)
                        \using [T]
                \endprooftree
                \justifies
                N(A)\tto T(A),N(B)\vdash T(B)
                \using \tto_l
        \endprooftree
        \justifies
        N(A)\tto T(A)\vdash N(B)\tto T(B)
        \using \tto_r
\endprooftree
\]\vspace{5mm}
\[
\prooftree
        \prooftree
                \pi
                \justifies
                A\vdash B
        \endprooftree
        \justifies
        \mu Y.T(A) \vdash \mu Y.T(B)
        \using [\mu Y.T]
\endprooftree
\!\!\leadsto\!
\prooftree
        \prooftree
                \prooftree
                        \prooftree
                                \pi
                                \justifies
                                A\vdash B
                        \endprooftree
                        \justifies
                        T(A)[\mu Y.T(B)/Y]\vdash T(B)[\mu Y.T(B)/Y]
                        \using [T[\mu Y.T(B)/Y]]
                \endprooftree
                \justifies
                T(A)[\mu Y.T(B)/Y]\vdash \mu Y.T(B)
                \using \mu_r
        \endprooftree
        \justifies
        \mu Y.T(A)\vdash \mu Y.T(B)
        \using \mu_l
\endprooftree
\]
}

\normalsize{We omit the rule for $\nu$, which is dual, and for $\wedge$ and $\vee$, which are simple pairing and case manipulations. 
Note also that most of these cases have a counterpart where $T$ is replaced by negative $N$, which has the sole effect of
$\pi$ being a proof of $B\vdash A$ instead of $A\vdash B$ in the expansion rules. With that, we can express
the cut elimination rule for fixpoints:}

{\small{
\[
\hspace{-35pt}
\prooftree
        \prooftree
                \prooftree
                        \pi_1
                        \justifies
                        \Gamma \vdash T[\mu X.T/X]
                \endprooftree
                \justifies
                \Gamma\vdash \mu X.T
                \using \mu_r
        \endprooftree
        \prooftree
                \prooftree
                        \pi_2
                        \justifies
                        T[A/X]\vdash A
                \endprooftree
                \justifies
                \mu X. T \vdash A
                \using \mu_l
        \endprooftree
\justifies 
\Gamma\vdash A
\using Cut
\endprooftree
~~~~
\leadsto
\]
\[
\prooftree
        \prooftree
                \prooftree
                        \pi_1
                        \justifies
                        \Gamma\vdash T[\mu X.T/X]
                \endprooftree
                \prooftree
                        \prooftree
                                \prooftree
                                        \pi_2
                                        \justifies
                                        T[A/X]\vdash A
                                \endprooftree
                                \justifies
                                \mu X.T \vdash A
                                \using \mu_l
                        \endprooftree
                        \justifies
                        T[\mu X.T/X]\vdash T[A/X]
                        \using [T]
                \endprooftree
        \justifies 
        \Gamma \vdash T[A/X]
        \using Cut
        \endprooftree
        \prooftree
                \pi_2
                \justifies
                T[A/X]\vdash A
        \endprooftree
        \justifies
        \Gamma\vdash A
        \using Cut
\endprooftree
\]}}

\normalsize{}We skip once again the rule for $\nu$, which is dual to $\mu$. We choose consciously not to recall the usual cut 
elimination rules nor the associated commutation
rules, since they are not central to our goals. $\mu LJ$, as presented above, does not formally eliminate cuts since there is 
no rule to reduce the following (and its dual with $\nu$):
\[
\prooftree
	\prooftree
		\prooftree
			\pi_1
			\justifies
			T(A)\vdash A
		\endprooftree
		\justifies
		\mu X.T\vdash A
		\using \mu_l
	\endprooftree
	~~~~~
	\prooftree
		\pi_2
		\justifies
		\Gamma,A\vdash B
	\endprooftree
	\justifies
	\Gamma,\mu X.T\vdash B
	\using Cut
\endprooftree
\]
This cannot be reduced without some prior unfolding of the $\mu X.T$ on the left. This issue is often solved \cite{baelde2007lag}
by replacing the rule for $\mu$ presented here above by the following:
\[
\prooftree
	T(A)\vdash A~~~~~\Gamma,A\vdash B
	\justifies
	\Gamma,\mu X.T \vdash B
	\using \mu'	
\endprooftree
\]
With the corresponding reduction rule, and analogously for $\nu$. We choose here not to do this, first because our game
model will prove consistency without the need to prove cut elimination, and second because we want to preserve the proximity with the categorical
structure of initial algebras / terminal coalgebras. 

\subsection{The games model}

We present the game model for fixpoints. We wish to model a proof system, therefore we need
our strategies to be total. The base arenas of the interpretation of fixpoints will be 
the arenas with loops presented in section \ref{loops}, to which we will adjoin a winning
function. While the base arenas will be the same for greatest and least fixpoints, 
they will be distinguished by the winning function: intuitively, Player loses if 
a play grows infinite in a least fixpoint (inductive) game, and Opponent loses if this
happens in a greatest fixpoint (coinductive) game. The winning functions we are going
to present are strongly influenced by Santocanale's work on games for 
$\mu$-lattices \cite{santocanale2002fmul}.
A \textbf{win open functor} is a functor $\game{T}:(\Games\times \Games^{op})^n\to \Games$ such that there is an open functor 
$\arena{T}[X_1,\dots,X_n]$
such that for all games $\game{A}_1,\dots,\game{A}_{2n}$ of base arenas $\arena{A}_1,\dots,\arena{A}_{2n}$, the base arena of 
$\game{T}(\game{A}_1,\dots,\game{A}_{2n})$ is $\arena{T}(\arena{A}_1,\dots,\arena{A}_n)$. In other terms, it is the natural 
lifting of open functors to the category of games. By abuse of notation, we denote this by $\game{T}[X_1,\dots,X_n]$, and 
$\arena{T}[X_1,\dots,X_n]$ will denote its underlying open functor.

\subsubsection{Least fixed point.} Let $\game{T}[X_1,\dots,X_n]$ be a win open functor such that $\hhole_{X_1}$ appears only positively
and at depth higher than $0$ in $\arena{T}[X_1,\dots,X_n]$. Then we define a new win open functor $\mu X_1.\game{T}[X_2,\dots,X_n]$
as follows:
\begin{itemize}
\item Its base arena is $\mu X_1. \arena{T}[X_2,\dots,X_n]$ ;
\item If $\game{A}_3,\dots,\game{A}_{2n}\in \Games$, $\mathcal{G}_{\mu X_1.\game{T}(\game{A}_3,\dots,\game{A}_{2n})}(s)=W$ iff
	\begin{itemize}
		\item There is $N\in \mathbb{N}$ such that no path of $s$ takes the external loop more that $N$ times, and ;
		\item $s$ is winning in the subgame inside the loop, or more formally:\\
		$\mathcal{G}_{\game{T}(\game{I},\game{I},\game{A}_3,\dots,\game{A}_{2n})}(s_{\restrict_{\game{T}(\game{I},\game{I},\game{A}_3,\dots,\game{A}_{2n})}})=W$.
	\end{itemize}
\end{itemize}

\subsubsection{Greatest fixed point.} Dually, if the same conditions are satisfied, we define the win open functor 
$\nu X_1.\game{T}[X_1,\dots,X_n]$ as follows:
\begin{itemize}
\item Its base arena is $\mu X_1. \arena{T}[X_2,\dots,X_n]$ ;
\item If $\game{A}_3,\dots,\game{A}_{2n}\in \Games$, $\mathcal{G}_{\nu X_1.\game{T}(\game{A}_3,\dots,\game{A}_{2n})}(s)=W$ iff
        \begin{itemize}
                \item For any $N\in \mathbb{N}$, there is a path of $s$ crossing the external loop more than $N$ times, or ;
                \item $s$ is winning in the subgame inside the loop, or more formally:\\
		$\mathcal{G}_{\game{T}(\game{I},\game{I},\game{A}_3,\dots,\game{A}_{2n})}(s_{\restrict_{\game{T}(\game{I},\game{I},\game{A}_3,\dots,\game{A}_{2n})}})=W$.
        \end{itemize}
\end{itemize}

It is straightforward to check that these are still functors, and in particular win open functors. There is one particular case 
that is worth noticing: if $\game{T}[X]$ has only one hole which appears only in positive position and at depth greater than $0$, then 
$\mu X.\game{T}$ is a constant functor, \emph{i.e.} a game. Moreover, theorem \ref{iso} implies that it is isomorphic in \Inn~
to $\game{T}(\mu X.\game{T})$. It is straightforward to check that this isomorphism $i_\game{T}: \game{T}(\mu X.\game{T}) \to \mu X.\game{T}$ 
is winning (it is nothing but the identity strategy), which shows that they are in fact isomorphic in \Games. Then, one can prove the following theorem:

\begin{thm}If $\game{T}[X]$ has only one hole which appears only in positive position and at depth greater than $0$, then 
the pair $(\mu X.\game{T},i_\game{T})$ defines an \textbf{initial algebra} for $\game{T}[X]$ and $(\nu X.\game{T},i^{-1}_\game{T})$
defines a \textbf{terminal coalgebra} for $\game{T}[X]$.
\end{thm}
\begin{proof}We give the proof for initial alebras, the second part being dual.
Let $(\game{A},\sigma)$ another algebra of $\game{T}[X]$. We need to show that there is a
unique $\sigma^\dagger:\mu X. \game{T} \tto \game{B}$ such that
\[
\xymatrix@R=15pt@C=15pt{
\game{T}(\mu X.\game{T})\ar[r]^{\game{T}(\sigma^\dagger)}\ar[d]_{i_\game{T}}&\game{T}(\game{B})\ar[d]^{\sigma}\\
\mu X.\game{T}\ar[r]_{\sigma^\dagger}&\game{B}
}
\]
commutes. The idea is to iterate $\sigma$:
\[
\xymatrix{
\dots\ar[r]^{\game{T}^3(\sigma)}&\game{T}^3(\game{B})\ar[r]^{\game{T}^2(\sigma)}&\game{T}^2(\game{B})\ar[r]^{\game{T}(\sigma)}&\game{T}(\game{B})\ar[r]^\sigma&\game{B}
}
\]
and somehow to take the limit. In fact we can give a direct definition of $\sigma^\dagger$:
\begin{eqnarray*}
\sigma^{(1)} &=& \sigma \\
\sigma^{(n+1)} &=& \game{T}^n(\sigma);\sigma^{(n)}\\
\sigma^\dagger &=& \{s\in \mathcal{L}_{\mu X.\game{T} \tto \game{B}}~|~\exists n\in \mathbb{N}^*,~ s\in \sigma^{(n)}\}
\end{eqnarray*}
This defines an innocent strategy, since when restricted to plays of $\mu X.\game{T}$, these strategies agree on their common domain.
This strategy is winning. Indeed, take an infinite play $\overline{s}\in \sigma^\dagger$. Suppose $\overline{s}_{\restrict_{\mu X. \game{T}}}$
is winning. By definition of $\mathcal{G}_{\mu X.\game{T}}$, this means that there is $N\in \mathbb{N}$ such that no path of 
$\overline{s}_{\restrict_{\mu X. \game{T}}}$ takes the external loop more than $N$ times. Thus, 
$\overline{s}\in \overline{L_{\game{T}^n(\game{I})\tto \game{B}}}$. But this implies that $\overline{s}\in \sigma^{(n)}$, and $\sigma^{(n)}$
is a composition of winning strategies thus winning, therefore $\overline{s}$ is winning.
Moreover, $\sigma^\dagger$ is the unique innocent strategy making the diagram commute: suppose there is another $f$ making this square commute. 
Since $\game{T}(\mu X.\game{T})$ and $\mu X.\game{T}$ have
the same set of paths, $i_\game{T}$ is in fact the identity, thus we have $\game{T}(f);\sigma = f$.
By applying $T$ and post-composing by $\sigma$, we get:
\[
\game{T}^2(f);\game{T}(\sigma);\sigma = \game{T}(f);\sigma = f
\]
And by iterating this process, we get for all $n\in \mathbb{N}$:
\[
\game{T}^{n+1}(f);\game{T}^n(\sigma);\dots;\game{T}(\sigma);\sigma = f
\]
Thus:
\[
\game{T}^{n+1}(f);\sigma^{(n)} = f
\]
Now take $s\in f$, and let $n$ be the length of the longest path in $s$. Since $\game{T}[X]$ has no hole at the root,
no path of length $n$ can reach $B$ in $\game{T}^{n+1}(B)$, thus $s\in \sigma^{(n)}$, therefore
$s\in \sigma^\dagger$. The same reasoning also works for the other inclusion.
Likewise, if $\sigma: \game{B}\to \game{T}(\game{B})$, we build a unique $\sigma^\ddagger: \game{B}\to \nu X.\game{T}$ making
the coalgebra diagram commute.
\end{proof}

\subsection{Interpretation of $\mu LJ$}

\subsubsection{Interpretation of Formulas.} As expected, we give the interpretation of valid formulas.
\[
\begin{array}{rclcrcl}
\intr{\top} &=& \game{I}			&~~~~~~	&\intr{A\tto B} &=& \intr{A}\tto \intr{B}\\
\intr{\bot} &=& \game{\bot}			&~~~~~~	&\intr{X} &=& \hhole_X\\
\intr{A\vee B} &=& \intr{A} + \intr{B}		&~~~~~~	&\intr{\mu X. T} &=& \mu X.\intr{T}\\
\intr{A\wedge B} &=& \intr{A} \times \intr{B}	&~~~~~~	&\intr{\nu X. T} &=& \nu X. \intr{T}
\end{array}
\]

\subsubsection{Interpretation of Proofs.} As usual, the interpretation of a proof $\pi$ of a sequent $A_1,\dots,A_n\vdash B$ will 
be a morphism $\intr{\pi}: \intr{A_1}\times \dots\times \intr{A_n} \longrightarrow \intr{B}$. The interpretation is computed
by induction on the proof tree. The interpretation of the rules of LJ is standard and its correctness follows from the cartesian
closed structure of \Games. Here are the interpretations for the fixpoint and functor rules:

{
\small{
\[
\bigintr{
\prooftree
        \prooftree
                \pi
                \justifies
                \Gamma\vdash T[\mu X.T/X]
        \endprooftree
        \justifies
        \Gamma\vdash \mu X.T
        \using \mu_r
\endprooftree
} = \intr{\pi};i_{\intr{T}}
~~~~~~~~~~~~~~~~~
\bigintr{
\prooftree
        \prooftree
                \pi
                \justifies
                T[A/X]\vdash A
        \endprooftree
        \justifies
        \mu X.T\vdash A
        \using \mu_l
\endprooftree
} = \intr{\pi}^\dagger
\]
\[
\bigintr{
\prooftree
        \prooftree
                \pi
                \justifies
                T[\nu X.T/X]\vdash B
                \endprooftree
        \justifies
        \nu X.T\vdash B
        \using \nu_l
\endprooftree
} = i_{\intr{T}}^{-1};\intr{\pi}
~~~~~~~~~~~~~~
\bigintr{
\prooftree
        \prooftree
                \pi
                \justifies 
                A\vdash T[A/X]
                \endprooftree
        \justifies
        A\vdash \nu X.T
        \using \nu_r
\endprooftree
} = \intr{\pi}^\ddagger
\]
\[
\bigintr{
\prooftree
        \prooftree
                \pi
                \justifies
                A\vdash B
        \endprooftree
        \justifies T(A)\vdash T(B)
        \using [T]
\endprooftree
} = \intr{T}(\intr{\pi})
\]}
}

We do not give the details of the proof that this defines an invariant of reduction. The main technical point is the validity 
of the interpretation of the functor rule; more precisely when the functor is a (least or greatest) fixpoint.
Given that, we get the following theorem.
\begin{thm}
If $\pi \leadsto \pi'$, then $\intr{\pi} = \intr{\pi'}$.
\end{thm}
In particular, this proves the following theorem which is certainly worth noticing, because $\mu LJ$ has large expressive power. 
In particular, it contains G\"odel's system T \cite{godel1958ubn}.
\begin{thm}
$\mu LJ$ is consistent: there is no proof of $\bot$.
\end{thm}
\begin{proof}
There is no total strategy on the game $\game{\bot}$.
\end{proof}

\subsubsection{Completeness.}
When it comes to completeness, we run into the issue that the total winning innocent strategies are not necessarily finite,
hence the usual definability process does not terminate. Nonetheless, we get a definability theorem in an infinitary version
of $\mu LJ$.
Whether a more precise completeness theorem is possible is a subtle point. First, we would need to restrict to an adequate subclass of
the recursive total winning strategies (for example, the Ackermann function is definable in $\mu LJ$). Then again, the problem to find a 
proof whose interpretation is \emph{exactly} the original strategy would be highly non-trivial: if $\sigma: \mu X.T\tto A$, we have to 
\emph{guess} an invariant $B$, a proof $\pi_1$ of $T(B)\vdash B$ and a proof $\pi_2$ of $B\vdash A$ such that 
$\intr{\pi_1}^\dagger;\intr{\pi_2} = \sigma$. Perhaps it would be more feasible to look for a proof whose interpretation is 
\emph{observationally equivalent} to the original strategy, which would be very similar to the universality result in \cite{hyland2000fap}.

\section{Conclusion and Future Work}

We have successfully constructed a games model of a propositional intuitionistic sequent calculus $\mu LJ$ with inductive and coinductive
types. It is striking that the adequate winning conditions on legal plays to model (co)induction are almost identical to those used
in parity games to model least and greatest fixpoints, to the extent that the restriction of our winning condition to paths coincides
\emph{exactly} with the winning condition used in \cite{santocanale2002fmul}. It would be worthwile to investigate this connection further:
given a game viewed as a bipartite graph along with winning conditions for infinite plays, under which assumptions can these winning conditions
be canonically lifted to the set of legal plays on this graph, viewed as an arena? Results in this direction might prove useful, since
they would allow to import many game-theoretic results into game semantics, and thus programming languages.

This work is part of a larger project to provide game-theoretic models to total programming languages with dependent types, such as
COQ or Agda. In these settings, (co)induction is crucial, since they deliberately lack general recursion. We believe that in the appropriate
games setting, we can push the present results further and model Dybjer's Inductive-Recursive\cite{dybjer2000gfs} definitions.

\subsubsection{Acknowledgements.} We would like to thank Russ Harmer, Stephane Gimenez and David Baelde for stimulating discussions, 
and the anonymous referees for useful comments and suggestions.

\bibliographystyle{plain}
\bibliography{fixpoints}
\end{document}